\documentclass[conference]{IEEEtran}
\IEEEoverridecommandlockouts
\usepackage{times}
\usepackage{cite}
\usepackage{flushend}
\usepackage[dvips]{color}
\usepackage{epsf}
\usepackage{epsfig}
\usepackage{graphicx}
\usepackage{graphics}
\usepackage{epstopdf}
\usepackage[small]{caption}
\usepackage[bottom]{footmisc}
\usepackage{amsmath}
\usepackage{amssymb}
\usepackage{amsthm}
\usepackage{amsxtra}
\usepackage{algorithm}
\usepackage{algorithmic}
\usepackage{enumerate}
\usepackage{multirow}
\usepackage{enumerate}
\usepackage{amssymb}
\usepackage{lipsum}
\usepackage{setspace}
\usepackage{multirow}
\usepackage{tikz}
\usepackage{setspace}
\usepackage{subfigure}

\newtheorem{theorem}{\bf Theorem}

\newtheorem{definition}{\bf Definition}

\usepackage{rotating} 
\usepackage{graphicx} 

\newcommand{\Rmnum}[1]{\expandafter\@slowromancap\romannumeral #1@}

\makeatother
\begin{document}

\title{Energy Management for a User Interactive Smart Community: A Stackelberg Game Approach}
\author{\IEEEauthorblockN{Wayes Tushar\IEEEauthorrefmark{1},
Bo Chai\IEEEauthorrefmark{2}, Chau Yuen\IEEEauthorrefmark{1}, David B. Smith\IEEEauthorrefmark{3} and  H. Vincent Poor\IEEEauthorrefmark{4}}
\IEEEauthorblockA{\IEEEauthorrefmark{1}Singapore University of Technology and Design, Singapore 138682. Email: \{wayes\_tushar, yuenchau\}@sutd.edu.sg. \\\IEEEauthorrefmark{2}State Key Lab. of Industrial Control Technology, Zhejiang University, China. Email: chaibozju@gmail.com.\\
\IEEEauthorrefmark{3}NICTA, Canberra, ACT, Australia. Email: david.smith@nicta.com.au.\\
\IEEEauthorrefmark{4}School of Engineering and Applied Science, Princeton University, Princeton, NJ, USA. Email: poor@princeton.edu.}
\thanks{\IEEEauthorrefmark{1}This work is supported by the Singapore University of Technology and Design (SUTD) under Energy Innovation Research Program (EIRP) Singapore NRF2012EWT-EIRP002-045.}
\thanks{\IEEEauthorrefmark{3}David Smith is also with the Australian National University (ANU), and his work is supported by NICTA. NICTA is funded by the Australian Government through the Department of Communications and the Australian Research Council through the ICT Centre of Excellence Program.}
}
\maketitle
\begin{abstract}
This paper studies a three party energy management problem in a user interactive smart community that consists of a large number of residential units (RUs) with distributed energy resources (DERs), a shared facility controller (SFC) and the main grid. A Stackelberg game is formulated to benefit both the SFC and RUs, in terms of incurred cost and achieved utility respectively, from their energy trading with each other and the grid. The properties of the game are studied and it is shown that there exists a unique Stackelberg equilibrium (SE). A novel algorithm is proposed that can be implemented in a distributed fashion by both RUs and the SFC to reach the SE.  The convergence of the algorithm is also proven, and shown to always reach the SE. Numerical examples are used to assess the properties and effectiveness of the proposed scheme.
\end{abstract}
\begin{IEEEkeywords}
Smart grid, distributed energy resources, game theory, energy management.
\end{IEEEkeywords}
\IEEEpeerreviewmaketitle
\section{Introduction}\label{introduction}
Distributed energy resources (DERs) have the capability of assisting consumers is reducing their dependence on the main grid as their primary source of electricity, and thus, lowering their costs of energy purchase~\cite{Tham-JTSMCS:2013}. They are also critical to the reduction of green house emissions and alleviation of climate change~\cite{Georgilakis-JTPS:2013}. As a result, there has been an increasing interest in deploying DERs in the smart grid. The majority of recent works in managing energy using DERs have mainly focussed on two areas: 1) the study of feasibility and control of DERs for their use in designing efficient micro-grids, e.g., see \cite{Justo-J-RSER:2013} and the references therein; and 2) scheduling energy consumption of household equipment by exploiting the use of DERs to optimize different grid operational objectives such as minimizing the energy consumption costs of users~\cite{Zhang-J_ECM:2013, Chai-ASCC:2013}. In most cases it is assumed that the users with DERs also possess storage devices. However, there are also some cases in which users might not want to store energy. Rather, they are more inclined to consume or trade energy as soon as it is generated, e.g., as in a grid-tie solar system without battery back up~\cite{wt_battery_solar:2013}. Furthermore, the majority of research on energy management emphasizes energy trading between two energy entities, i.e., two-way energy flow. For example, a considerable number of references that use such models can be found in \cite{Fang-J-CST:2012, Hassan-Energies:2013,Yu-IEEENetworks:2011, Liu-ISGT:2013, Hassan-ISGT:2013}.

In this paper, a three party energy management scheme is proposed for a smart community that consists of multiple residential units (RUs), a shared facility controller (SFC) and the main grid. To the best of our knowledge, this paper is the first that introduces the idea of a shared facility and considers a 3-party energy management problem in smart grid. With the development of modern residential communities, shared facilities provide essential public services to the RUs, e.g., maintenance of lifts in community apartments. Hence, it is necessary to study the energy demand management of shared facilities for expediting effective community work. In particular, for the considered setting, as will be seen shortly, energy trading of RUs with the grid and the SFC constitutes an important energy management problem for both the SFC and RUs. On the one hand, each RU is interested in selling its energy either to the SFC or to the grid at a higher price to increase revenue. On the other hand, the SFC wants to minimize its cost of energy purchased by making a price offer to RUs to encourage them to sell their energy to the SFC instead of the grid. This enables the SFC to be less dependent on expensive electricity from the grid.

As an energy management tool, the framework of a noncooperative Stackelberg game (NSG)~\cite{Maharjan-JTSG:2013} is considered. In fact, NSGs have been used extensively in designing different energy management solutions. For example, maximizing revenues of multiple utility companies and customers~\cite{Maharjan-JTSG:2013, Chai-TSG:2014}, minimizing customers' bills to retailers while maximizing retailers' profits~\cite{Meng-JSpringer:2013}, prioritizing consumers' interests in designing energy management solutions~\cite{Tushar-TSG:2013}, and managing energy between multiple micro-grids in the smart grid~\cite{Asimakopoulou-JTSG:2013}, among many others. However, the choice of players and their strategies significantly differ from one game to another based on the system model, the objective of energy management design and the use of algorithms. To that end, an NSG  is proposed for the considered scenario to capture the interaction between the SFC and RUs and it is shown that the maximum benefits to the SFC and RUs are achieved at the SE of the game. The properties of the game are studied, and it is proven that there exists a unique SE. Finally, a novel algorithm, which is guaranteed to reach the SE, and can be implemented in a distributed fashion by the SFC and the RUs is introduced. The effectiveness of the proposed scheme is confirmed by numerical simulations. 
\section{System Model}\label{system-model}
\begin{figure}[t!]
\centering
\includegraphics[width=\columnwidth]{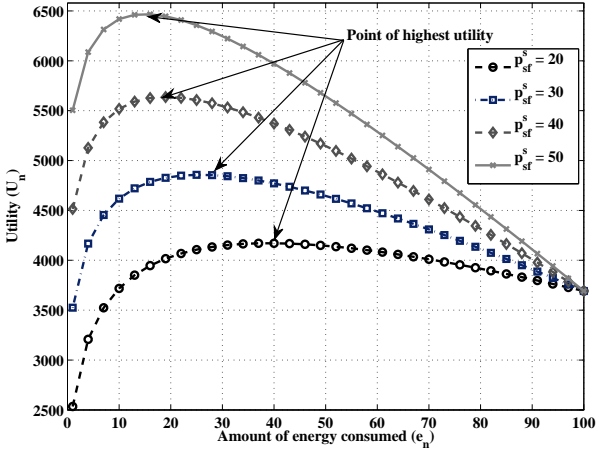}
\caption{Shift of maximum utility point as the price per unit of energy set by the SFC to pay to each RU increases. A higher price results in less consumption by the RU and vice versa.} \label{fig:utilityVsPrice}
\end{figure}
Consider a smart grid network consisting of the main grid and a smart community with $N$ RUs and an SFC, which are connected to one another via communication and power lines. Each RU, which is equipped with DERs such as solar panels or wind turbines, can be a single residential unit or group of units connected via an  aggregator that acts as a single entity. All RUs are considered to belong to the set $\mathcal{N}$. Here, on the one hand, the SFC does not have any electricity generation capacity.  Hence, at any time of the day, it needs to rely on the grid and RUs for required energy $E_\text{sf}^\text{req}$ to run equipment and machines in the shared facility such as lifts, water pumps, parking gates and lights that are shared and used by the residences on daily basis. On the other hand, each RU $n\in\mathcal{N}$ is considered to have no storage capability, and therefore, wants to consume or sell its generated energy $E_n^\text{gen}$ either to the main grid or to the SFC to raise revenue. It is assumed that each RU can manage its consumption $e_n$, and thus sell the rest of the generated energy $E_n^\text{gen} - e_n$ to the SFC or to the grid. Clearly, if $E_n^\text{gen}\leq E_n^\text{min}$, where $E_n^\text{min}$ is the base load for RU $n$, the RU cannot take part in the energy management. Otherwise, which is the considered case, the RU sells $E_n^\text{gen} - e_n$ after controlling its consumption amount $e_n$.

In general, the buying price $p_g^b$ of a grid is noticeably lower than its selling price $p_g^s$~\cite{McKenna-JIET:2013}. To this end, it is assumed that the price $p_\text{sf}^s$ per unit of energy that the SFC pays to each RU is set between the buying and selling price of the grid. Therefore, each RU can sell its energy at a higher price $p_\text{sf}^s>p_g^b$ and the SFC can buy at a lower price $p_\text{sf}^s<p_g^s$ by trading energy among themselves rather than trading with the grid. Under this condition, it is reasonable to assume that the RU $n$ would be more inclined to sell $E_n^\text{gen} - e_n$ to the SFC instead of to the grid. To that end, the amount of utility that an RU achieves from its energy consumption $e_n$ and trading the rest with the SFC can be modeled as 
\begin{eqnarray}
U_n = k_n\ln(1+e_n) + p_\text{sf}^s(E_n^\text{gen} - e_n), k_n>0.
\label{eqn:1}
\end{eqnarray}
In \eqref{eqn:1}, $k_n\ln(1+e_n)$ is the utility that the RU $n$ achieves from consuming $e_n$, and $k_n$ is a preference parameter~\cite{Samadi-C-Smartgridcomm:2010}. $p_\text{sf}^s(E_n^\text{gen} - e_n)$ is the revenue that the RU receives from selling the rest of its energy to the SFC. Please note that the natural logarithm $\ln(\cdot)$ has been used extensively for utility functions~\cite{Pavlidou-JCN:2008}, and has particularly been shown to be suitable for modeling the utility for power consumers~\cite{Maharjan-JTSG:2013}. From \eqref{eqn:1}, the RU $n$ would be interested in selling more energy to the SFC, e.g., by scheduling its use of devices at a later time, if the values of $k_n$ and $p_\text{sf}^s$ are high and vice-versa. The effect of $p_\text{sf}^s$ on the achieved utility by an RU is illustrated in Fig.~\ref{fig:utilityVsPrice}. The figure clearly shows that at a higher $p_\text{sf}^s$ maximum utility is achieved by an RU when it consumes less, i.e., it sells more to the SFC.

On the other hand, the SFC buys all its required energy $E_\text{sf}^\text{req}$ from RUs and the grid. Due to the choice of price $p_\text{sf}^s$, i.e., $p_\text{sf}^s<p_g^s$, the SFC is more interested in buying its energy from RUs and then procuring the rest, if there is any, from the grid at a price $p_g^s$. To this end, a cost function for the SFC is defined as
\begin{eqnarray}
J_\text{sf} = p_\text{sf}^s\sum_n e_\text{n,sf}^s + (E_\text{sf}^\text{req} - \sum_n e_\text{n,sf}^s)p_g^s\label{eqn:2}
\end{eqnarray}
to capture its total cost of buying energy from RUs and the grid. In \eqref{eqn:2}, $e_\text{n,sf}^s = E_n^\text{gen} - e_n$ is the amount of energy that the SFC buys from RU $n$. Now if $p_\text{sf}^s$ is too low it might cause an RU to refrain from selling its energy to the SFC. As a result, the SFC would need to buy all its $E_\text{sf}^\text{req}$ from the grid at a higher rate. On the contrary, if $p_\text{sf}^s$ is very high, it will increase the cost to the SFC significantly. Hence, $p_\text{sf}^s$ should be within a legitimate range to encourage the RUs to sell their energy to the SFC, while at the same time, keeping the cost to the SFC at a minimum. 

Now, to decide on the energy trading parameters $e_{n,\text{sf}}^s$ and $p_\text{sf}^s$, on the one hand, the SFC interacts with each RU $n\in\mathcal{N}$ to minimize \eqref{eqn:2} by choosing a suitable price to pay to each $n$. On the other hand, each RU decides on the amount of energy $e_n$ that it wants to consume and thus maximize \eqref{eqn:1}. To capture this interaction, an NSG between the SFC and RUs is proposed in the next section.

\section{Noncooperative Stackelberg Game and Its Properties}\label{sec:game-formulation}
First, the objective of each RU is to decide on the amount of energy $e_n$ that it wants to consume, and thus to determine $e_\text{n,sf}^s$ based on the offered price $p_\text{sf}^s$ to sell to the SFC such that \eqref{eqn:1} possesses the maximum value. Mathematically,
\begin{eqnarray}
\max_{e_n} \left[k_n\ln(1+e_n) + p_\text{sf}^s(E_n^\text{gen} - e_n)\right].\label{eqn:3}
\end{eqnarray} 
Conversely, having the offered energy from all RUs, i.e., $e_\text{n,sf}^s~\forall n$, the SFC determines the price $p_\text{sf}^s$ so as to minimize the cost captured via \eqref{eqn:2}. Therefore, the objective of the SFC is 
\begin{eqnarray}
\min_{p_\text{sf}^s}\left[ p_\text{sf}^s\sum_n e_\text{n,sf}^s + (E_\text{sf}^\text{req} - \sum_n e_\text{n,sf}^s)p_g^s\right].\label{eqn:4}
\end{eqnarray}
Here, \eqref{eqn:3} and \eqref{eqn:4} are concave and convex functions respectively, and are coupled via common parameters $e_n$ and $p_\text{sf}^s$. Therefore, it would be possible to solve the problem in an optimal centralized fashion if private information such as $k_n$ and $E_n^\text{gen}$ were available to the central controller. However, to protect the privacy of each RU as well as to reduce the demand on communications bandwidth, it is useful to develop a distributed mechanism. With these considerations in mind, we study the problem using an NSG.
\subsection{Noncooperative Stackelberg Game}\label{sec:nsg}
 A Stackelberg game, also known as a leader-follower game, studies the multi-level decision making processes of a number of independent players, i.e., followers, in response to the decision made by the leader (or, leaders) of the game~\cite{Maharjan-JTSG:2013}. In the proposed NSG, the SFC and each RU are modeled as the leader and a follower respectively. Formally, the NSG can be defined by its strategic form as
 \begin{eqnarray}
 \Gamma = \{(\mathcal{N}\cup\{\text{SFC}\}), \{\mathbf{E}_{n\in\mathcal{N}}\}, \{U_n\}_{n\in\mathcal{N}}, p_\text{sf}^s, J_\text{sf}\},\label{eqn:5}
 \end{eqnarray}
which has following components:
 \begin{enumerate}[i)]
 \item The set $\mathcal{N}$ of all followers in the game.
 \item The set $\{\text{SFC}\}$ of leaders in the game that has only one element in our case, i.e., a single leader.
 \item The strategy set $\mathbf{E}_n$ of each RU $n\in\mathcal{N}$ to choose an amount of energy $e_n\in\mathbf{E}_n$ to be consumed during the game.
 \item The utility function $U_n$ of each RU $n$ to capture the benefit from consuming $e_n$, and the utility from selling $e_\text{n,sf}^s = E_n^\text{gen} - e_n$ to the SFC.
 \item The price $p_\text{sf}^s$ set by the SFC to buy its energy from RUs.
 \item The cost function $J_\text{sf}$ of the SFC that quantifies the total cost of energy purchase from RUs and the grid.
 \end{enumerate}
 Through $\Gamma$, all RUs that want to trade their energy and the SFC interact with each other and decide on the decision vector $\mathbf{e} = [e_1, e_2, \hdots, e_n, \hdots, e_N]$ and $p_\text{sf}^s$ by choosing their appropriate strategies. In this regard, one suitable solution of the proposed $\Gamma$ is the SE, which is obtained as soon as the leader decides on its optimal price based on the followers' best responses of their offered energy. 
 \begin{definition}
Consider the NSG $\Gamma$ as defined by \eqref{eqn:5} where $U_n$ and $J_\text{sf}$ are determined by \eqref{eqn:1} and \eqref{eqn:2} respectively. A set of strategies $(\mathbf{e}^*, p_\text{sf}^{s^*})$ comprises the SE of the proposed $\Gamma$ if it satisfies the following set of inequalities:
\begin{eqnarray}
U_n(\mathbf{e}^*, p_\text{sf}^{s^*})\geq U_n(e_n, {\mathbf{e}_{-n}^*}, p_\text{sf}^{s^*}), \forall n\in\mathcal{N}, e_n\in\mathcal{N},
\label{eqn:6}
\end{eqnarray}
and
\begin{eqnarray}
J_\text{sf}(\mathbf{e}^*, p_\text{sf}^{s^*})\leq J_\text{sf}(\mathbf{e}^*, p_\text{sf}^{s}),\label{eqn:7}
\end{eqnarray}
where $\mathbf{e}_{-n}$ is the strategy set of all RUs in $\mathcal{N}/\{n\}$.
\label{def:1}
 \end{definition}
Therefore, according to \eqref{eqn:6} and \eqref{eqn:7}, neither the SFC nor any RU in the set $(\mathcal{N}\cup\{\text{SFC}\})$ can benefit, in terms of its total cost and achieved utility respectively, by unilaterally changing its strategy once the NSG $\Gamma$ reaches an SE. 
\subsection{Existence and Uniqueness of SE}
The existence of a pure strategy solution is not always guarateed in noncooperative games~\cite{Maharjan-JTSG:2013}. Hence, there is a need to investigate whether there exists any SE for the proposed NSG. The following theorem settles this issue.
\begin{theorem}
There exists a unique pure strategy SE in the proposed NSG $\Gamma$ between the SFC and RUs in the set $(\mathcal{N}\cup\{\text{SFC}\})$.\label{thm:1}
\end{theorem}
\begin{proof}
First, note that $U_n$ in \eqref{eqn:1} is a strictly concave function of $e_n,~\forall n\in\mathcal{N}$, i.e., $\frac{\delta^2 U_n}{\delta e_n^2}<0$. Therefore, for any price $p_\text{sf}^s>0$, each RU $n$ will have a unique $e_n$, chosen from a bounded strategy set $[E_n^\text{min}, E_n^\text{gen}]$\footnote{An RU must consume at least its base load, and cannot consume more than its generation, at any time.}, that maximizes $U_n$. It is also noted that $\Gamma$ reaches SE when all players including the SFC and each RU $n\in\mathcal{N}$ have their best cost and utilities respectively with respect to the strategies chosen by all players in the game. Thereby, it is indisputable that the proposed game $\Gamma$ would find an SE as soon as the SFC is able to find an optimal price $p_\text{sf}^{s^*}$ while all RUs play their unique strategy vector $\mathbf{e}^*$.

Now the second derivative of \eqref{eqn:2} with respect to $p_\text{sf}^s$ is
\begin{eqnarray}
\frac{\delta^2 J_\text{sf}}{\delta p_\text{sf}^{s^2}} = \frac{2\sum_n k_n}{(p_\text{sf}^s)^3},
\end{eqnarray} 
which is greater than $0$. Therefore, $J_\text{sf}$ is strictly convex with respect to $p_\text{sf}^s$. Consequently, the SFC is able to find a unique price $p_\text{sf}^{s^*}$ in response to the strategy vector $\mathbf{e}^*$. Thus, there exists a unique SE in the proposed NSG, and Theorem~\ref{thm:1} is proved.
\end{proof}
\subsection{Distributed Algorithm}
\begin{algorithm}[t]
\caption{Algorithm to reach the SE}
\label{alg:1}
\begin{algorithmic}[1]
\small
\STATE Initialization: $p_\text{sf}^{s^*}=0$ $J_\text{sf}^*=p_g^s*E_\text{sf}^\text{req}$
\FOR {Buying pricing $p_\text{sf}^s$  from $p_g^b$ to $p_g^s$ }
   \FOR {Each RU $n \in \mathcal{N}$}
        \STATE RU $n$ adjusts its energy consumption $e_n$ according to 
        \begin{equation}\label{eqn:alg-1}
           e_n^* = {\rm{arg}}{\kern 1pt} {\kern 1pt} {\kern 1pt} \mathop {\max }\limits_{0 \le {e_n} \le E_n^\text{gen}} {\kern 1pt} {\kern 1pt} {\kern 1pt} [{k_n}\ln (1 + {e_n}) + p_\text{sf}^s(E_n^\text{gen} - {e_n})].
        \end{equation} \\
   \ENDFOR
    \STATE The SFC computes the cost according to 
        \begin{equation}\label{eqn:alg-2}
    {J_\text{sf}} = p_\text{sf}^s\sum\limits_{n \in \mathcal{N}} {(E_n^\text{gen} - {e_n})}  + p_g^s\left(E_\text{sf}^\text{req} -
     \sum\limits_{n \in \mathcal{N}} {(E_n^\text{gen} - {e_n})} \right).
        \end{equation} \\
     \IF {$J_\text{sf} \le J_\text{sf}^*$}
      \STATE The SFC keeps records of the optimal price and minimal cost
         \begin{equation}\label{eqn:alg-3}
           p_\text{sf}^{s^*}=p_\text{sf}^{s}, J_\text{sf}^*=J_\text{sf}
         \end{equation}
     \ENDIF  
\ENDFOR\\
\textbf{The SE $(\mathbf{e}^*, p_\text{sf}^{s^*})$ is achieved.}
\end{algorithmic}
\end{algorithm}
In this section, an iterative algorithm that the SFC and RUs can implement in a distributed fashion is proposed to reach the SE of the game. In order to attain the unique SE, the SFC needs to communicate with each RU. At each iteration, on the one hand, the RU $n$ chooses its best energy consumption amount $e_n$ in response to the price $p_\text{sf}^s$ set by the SFC, calculates $e_{n,\text{sf}}^s= E_n^\text{gen} - e_n$ and sends it to the SFC. On the other hand, having the information on the choice of energy $e_{n,\text{sf}}^s~\forall n$, the SFC derives its price $p_\text{sf}^s$ to minimize its cost in \eqref{eqn:2} and resends it to each RU. The interaction between the SFC and all RUs continues iteratively until \eqref{eqn:6} and \eqref{eqn:7} are satisfied. As soon as these conditions are met, the proposed NSG reaches the SE. Details are given in Algorithm~\ref{alg:1}.
\begin{theorem}
The proposed Algorithm~\ref{alg:1} is always guaranteed to reach the SE of the game. 
\end{theorem}
\begin{proof}
According to the proposed algorithm, the conflict between RUs' choices of strategies stem from their impact on the choice of $p_\text{sf}^s$ by the SFC. Due to the strict convexity of $J_\text{sf}$, the choice of $p_\text{sf}^{s^*}>0$ lowers the cost of the SFC to the minimum. Now, as the algorithm is designed, in response to the $p_\text{sf}^{s^*}$, each RU $n$ chooses its strategy $e_n$ from the bounded range $\left[E_n^\text{min}, E_n^\text{gen}\right]$ to maximize its concave utility function $U_n$. Hence, due to a bounded strategy set and the continuity of the utility function $U_n$ with respect to $e_n$, each RU $n$ also reaches a fixed point at which its utility is maximized for the given price $p_\text{sf}^{s^*}$~\cite{Maharjan-JTSG:2013}. As a consequence, the proposed algorithm is always guaranteed to converge to the unique SE of the game.
\end{proof}
\section{Numerical Experiments}\label{numerical-simulation}
\begin{figure}[t]
\centering
\includegraphics[width=\columnwidth]{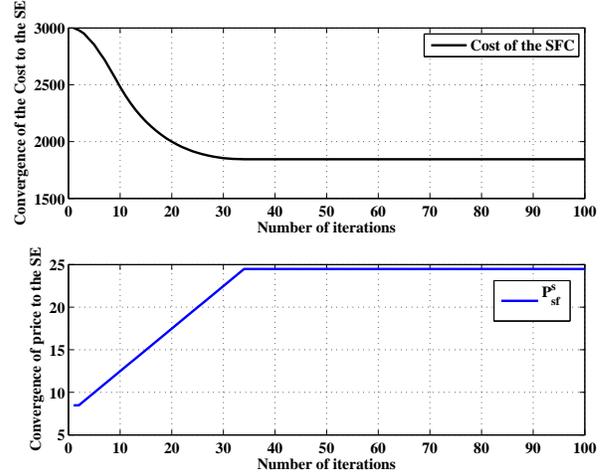}
\caption{Utility achieved by each RU at the SE.} \label{fig:convergence}
\end{figure}
The proposed energy management scheme is simulated by considering a number of RUs that are interested in selling their energy to the SFC. Typical energy generation of each RU from its DERs is assumed to be $10$ kWh~\cite{NREL_wind_generation:2009} and the required energy by the SFC is presumed to be $50$ kWh during the considered time. The preference parameter $k_n$ is chosen sufficiently large, e.g., $k_n$ is chosen from range $[90, 150]$ for this case study, such that $e_n$ and $p_\text{sf}^s$ in \eqref{eqn:1} are always positive. The grid's per unit selling price is assumed to be $60$ cents/kWh~\cite{Jin-J-TVT:2013} whereby the SFC sets its initial price equal to the grid's buying price of $8.45$ cents/kWh~\cite{Tushar-TSG:2013} to pay to each RU. Nonetheless, it is very important to highlight that all parameter values are particular to this study and may vary according to the need of the SFC, power generation of the grid and DERs, and the energy policy of a country.

In Fig.~\ref{fig:convergence}, the SFC's total cost is shown to converge to the SE by following Algorithm~\ref{alg:1} for a network with five RUs. It can be seen that although the SFC wants to minimize its total cost, it cannot do so with its initial choice of price for payment to the RUs. In fact, through interaction with each RU of the network the SFC eventually increases its price in each iteration to encourage the RUs to sell more, and consequently the cost continuously reduces. As can be seen from Fig.~\ref{fig:convergence}, the SFC's choice of equilibrium price and consequently also the minimum total cost reach their SE after the $34^\text{th}$ iteration.
\begin{figure}[t]
\centering
\includegraphics[width=\columnwidth]{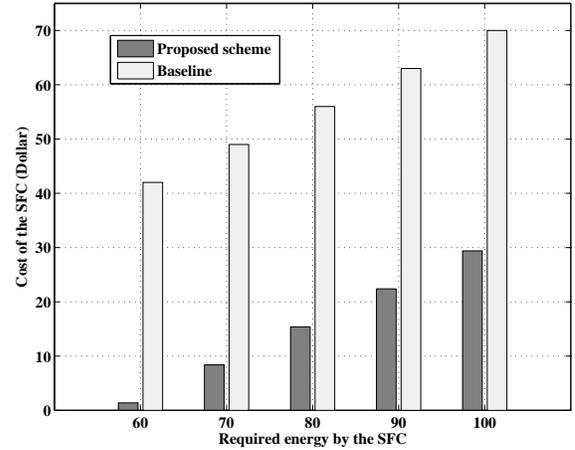}
\caption{Comparison of the cost incurred by the SFC between the proposed and baseline approaches for different amounts of required energy.} \label{fig:CostVsReq}
\end{figure}

Next, the effectiveness of the proposed scheme is demonstrated by comparing its performance with a standard baseline scheme that does not contain any DER facility, i.e., the SFC depends on the grid for all its energy. In this regard, considering $10$ RUs in the system, the total cost of energy trading that is incurred by the SFC is plotted in Fig.~\ref{fig:CostVsReq} for both the proposed and baseline approaches as the amount of energy required by the SFC increases. As shown in the figure, the cost to the SFC increases for both cases as the energy requirement increases from $60$ to $100$ kWh. In fact, it is a trivial result that a greater energy requirement leads the SFC to spend more money on buying energy, which consequently increases the cost. Nonetheless, the proposed scheme needs to spend significantly less to buy the same amount of energy due to the presence of the DERs of the RUs, and thus noticeably benefits from its energy trading in terms of total cost compared to the baseline scheme. As shown in Fig~\ref{fig:CostVsReq}, the SFC's cost is $74.9\%$, on average, lower than that of the baseline approach for the considered change in the SFC's energy requirement.
\begin{figure}[t]
\centering
\includegraphics[width=\columnwidth]{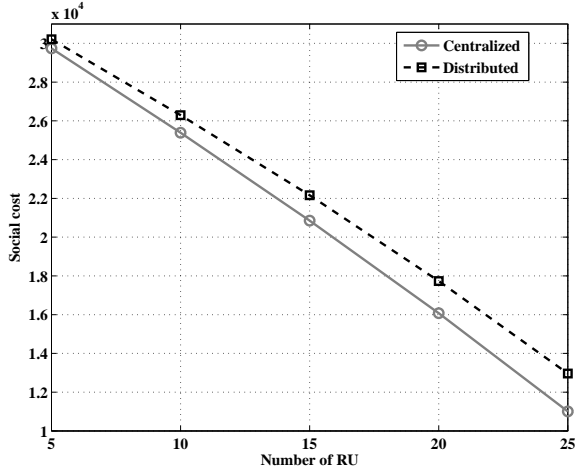}
\caption{Comparison of social cost obtained by the proposed distributed scheme with respect to the centralized scheme as the number of RUs varies in the network.} \label{fig:CentralVsPropose}
\end{figure}

Nevertheless, as mentioned in Section \ref{sec:game-formulation}, it is also possible to optimally manage energy between RUs and the SFC via a centralized control system to minimize the social cost\footnote{In contrast to social benefit, social cost is the difference between the total cost incurred by the SFC and total utility achieved by all RUs in the system.} if private information such as $k_n$ and $E_n^\text{gen}~\forall n$ is available to the controller. In this regard, the performance in terms of social cost for both the centralized and proposed distributed  schemes is observed in Fig.~\ref{fig:CentralVsPropose}. As can be seen from the figure, the social cost attained by adopting the distributed scheme is \emph{very close} to the optimal scheme at the SE of the game. However, the centralized scheme has access to the private information of each RU. Hence, the controller can optimally manage the energy, and as a result shows better performance in terms of reducing the SFC's cost compared to the proposed scheme. According to Fig.~\ref{fig:CentralVsPropose}, as the number of RUs changes in the network from $5$ to $25$, the average social cost for the proposed distributed scheme is only $7.07\%$ higher than that obtained via the centralized scheme. This is a promising result considering the distributed nature of the system.

\section{Conclusion}\label{conclusion}
In this paper, a user interactive energy management scheme has been proposed for a smart grid network that consists of a shared facility, the main grid and a large number of residential units (RUs). A noncooperative Stackelberg game (NSG) has been proposed that captures the interaction between the shared facility controller (SFC) and each RU and it has been shown to have a unique Stackelberg equilibrium (SE). It has been shown that the use of DERs for each RU is beneficial for both the SFC and RUs in terms of their incurred cost and achieved utilities respectively. Further, a distributed algorithm has been proposed, which is guaranteed to reach the SE and can be implemented by the players in a distributed fashion. Significant cost savings have been demonstrated for the SFC by comparing the proposed scheme with a standard baseline approach without any DERs.   

The proposed work can be extended in different directions. An interesting extension would be to examine the impact of discriminate pricing among the RUs on the outcome of the scheme. Another compelling augmentation would be to determine how to set the threshold on the grid's price. Further, quantifying the inconvenience that the SFC/RUs face during their interaction and quantifying the effect of the inclusion of storage devices could be other potential future extensions of the proposed work.  

\begin{thebibliography}{10}
\providecommand{\url}[1]{#1}
\csname url@samestyle\endcsname
\providecommand{\newblock}{\relax}
\providecommand{\bibinfo}[2]{#2}
\providecommand{\BIBentrySTDinterwordspacing}{\spaceskip=0pt\relax}
\providecommand{\BIBentryALTinterwordstretchfactor}{4}
\providecommand{\BIBentryALTinterwordspacing}{\spaceskip=\fontdimen2\font plus
\BIBentryALTinterwordstretchfactor\fontdimen3\font minus
  \fontdimen4\font\relax}
\providecommand{\BIBforeignlanguage}[2]{{%
\expandafter\ifx\csname l@#1\endcsname\relax
\typeout{** WARNING: IEEEtran.bst: No hyphenation pattern has been}%
\typeout{** loaded for the language `#1'. Using the pattern for}%
\typeout{** the default language instead.}%
\else
\language=\csname l@#1\endcsname
\fi
#2}}
\providecommand{\BIBdecl}{\relax}
\BIBdecl

\bibitem{Tham-JTSMCS:2013}
C.-K. Tham and T.~Luo, ``Sensing-driven energy purchasing in smart grid
  cyber-physical system,'' \emph{IEEE Transactions on Systems, Man, and
  Cybernetics: Systems}, vol.~43, no.~4, pp. 773--784, 2013.

\bibitem{Georgilakis-JTPS:2013}
P.~S. Georgilakis and N.~Hatziargyriou, ``Optimal distributed generation
  placement in power distribution networks: {M}odels, methods, and future
  research,'' \emph{IEEE Transactions on Power Systems}, vol.~28, no.~3, pp.
  3420--3428, 2013.

\bibitem{Justo-J-RSER:2013}
J.~J. Justo, F.~Mwasilu, J.~Lee, and J.-W. Jung, ``{AC}-microgrids versus
  {DC}-microgrids with distributed energy resources: {A} review,''
  \emph{Renewable and Sustainable Energy Reviews}, vol.~24, pp. 387--405, 2013.

\bibitem{Zhang-J_ECM:2013}
D.~Zhang, N.~Shah, and L.~G. Papageorgiou, ``Efficient energy consumption and
  operation management in a smart building with microgrid,'' \emph{Energy
  Conversion and Management}, vol.~74, pp. 209--222, 2013.

\bibitem{Chai-ASCC:2013}
B.~Chai, Z.~Yang, and J.~Chen, ``Optimal residential load scheduling in smart
  grid: {A} comprehensive approach,'' in \emph{Proc. IEEE Asian Control
  Conference (ASCC)}, Istanbul, Turkey, June 2013, pp. 1--6.

\bibitem{wt_battery_solar:2013}
SolarRay, ``Grid-tie package systems without batteries,'' website, 2012,
  \url{http://www.solarray.com/CompletePackages/Grid-Tie-No-BatteriesT.php}.

\bibitem{Fang-J-CST:2012}
X.~Fang, S.~Misra, G.~Xue, and D.~Yang, ``Smart grid - {T}he new and improved
  power grid: {A} survey,'' \emph{IEEE Communications Surveys Tutorials},
  vol.~14, no.~4, pp. 944--980, 2012.

\bibitem{Hassan-Energies:2013}
N.~U. Hassan, M.~A. Pasha, C.~Yuen, S.~Huang, and X.~Wang, ``Impact of
  scheduling flexibility on demand profile flatness and user inconvenience in
  residential smart grid system,'' \emph{Energies}, vol.~6, no.~12, pp.
  6608--6635, 2013.

\bibitem{Yu-IEEENetworks:2011}
R.~Yu, Y.~Zhang, S.~Gjessing, C.~Yuen, S.~Xie, and M.~Guizani, ``Cognitive
  radio based hierarchical communications infrastructure for smart grid,''
  \emph{IEEE Network}, vol.~25, no.~5, pp. 6--14, 2011.

\bibitem{Liu-ISGT:2013}
Y.~Liu, N.~U. Hassan, S.~Huang, and C.~Yuen, ``Electricity cost minimization
  for a residential smart grid with distributed generation and bidirectional
  power transactions,'' in \emph{Proc. IEEE PES Innovative Smart Grid
  Technologies (ISGT)}, Washington, DC, Feb 2013, pp. 1--6.

\bibitem{Hassan-ISGT:2013}
N.~U. Hassan, X.~Wang, S.~Huang, and C.~Yuen, ``Demand shaping to achieve
  steady electricity consumption with load balancing in a smart grid,'' in
  \emph{Proc. IEEE PES Innovative Smart Grid Technologies (ISGT)}, Washington,
  DC, Feb 2013, pp. 1--6.

\bibitem{Maharjan-JTSG:2013}
S.~Maharjan, Q.~Zhu, Y.~Zhang, S.~Gjessing, and T.~Ba{\c{s}}ar, ``Dependable
  demand response management in the smart grid: {A} {S}tackelberg game
  approach,'' \emph{IEEE Transactions on Smart Grid}, vol.~4, no.~1, pp.
  120--132, 2013.

\bibitem{Chai-TSG:2014}
B.~Chai, J.~Chen, Z.~Yang, and Y.~Zhang, ``Demand response management with
  multiple tility companies: A two-level game approach,'' \emph{IEEE
  Transactions on Smart Grid}, 2014, (to appear).

\bibitem{Meng-JSpringer:2013}
F.-L. Meng and X.-J. Zeng, ``A {S}tackelberg game-theoretic approach to optimal
  real-time pricing for the smart grid,'' \emph{Springer Soft Computing},
  vol.~17, no.~12, pp. 2365--2380, 2013.

\bibitem{Tushar-TSG:2013}
W.~Tushar, J.~A. Zhang, D.~Smith, H.~V. Poor, and S.~Thi{\'{e}}baux,
  ``Prioritizing consumers in smart grid: {A} game theoretic approach,''
  \emph{IEEE Transactions on Smart Grid}, 2013, (to appear).

\bibitem{Asimakopoulou-JTSG:2013}
G.~Asimakopoulou, A.~Dimeas, and N.~Hatziargyriou, ``Leader-follower strategies
  for energy management of multi-microgrids,'' \emph{IEEE Transactions on Smart
  Grid}, vol.~4, no.~4, pp. 1909--1916, 2013.

\bibitem{McKenna-JIET:2013}
E.~McKenna and M.Thomson, ``Photovoltaic metering configurations, feed-in
  tariffs and the variable effective electricity prices that result,''
  \emph{IET Renewable Power Generation}, vol.~7, no.~3, pp. 235--245, 2013.

\bibitem{Samadi-C-Smartgridcomm:2010}
P.~Samadi, A.-H. Mohsenian-Rad, R.~Schober, V.~Wong, and J.~Jatskevich,
  ``Optimal real-time pricing algorithm based on utility maximization for smart
  grid,'' in \emph{Proc. IEEE International Conference on Smart Grid
  Communications (SmartGridComm)}, Gaithersburg, MD, Oct. 2010, pp. 415--420.

\bibitem{Pavlidou-JCN:2008}
F.-N. Pavlidou and G.~Koltsidas, ``Game theory for routing modeling in
  communication networks: {A} survey,'' \emph{Journal of Communications and
  Networks}, vol.~10, no.~3, pp. 268--286, 2008.

\bibitem{NREL_wind_generation:2009}
T.~Forsyth, ``Small wind technology,'' website, National Renewable Energy
  Laboratory of US Department of Energy, 2009,
  \url{http://ww2.wapa.gov/sites/western/renewables/Documents/webcast\\/2SmallWindtech.pdf}.

\bibitem{Jin-J-TVT:2013}
C.~Jin, J.~Tang, and P.~Ghosh, ``Optimal electric vehicle charging: {A}
  customer's perspective,'' \emph{IEEE Transactions on Vehicular Technology},
  vol.~62, no.~7, pp. 2919--2927, 2013.

\end{thebibliography}

\end{document}